\documentclass[a4paper,UKenglish,cleveref, autoref, thm-restate]{lipics-v2021}

\usepackage{amsmath,amsthm,nicefrac}
\usepackage{amsfonts, amstext}
\usepackage{cases}
 \usepackage{libertine}
 \usepackage[libertine]{newtxmath}

\usepackage[font={small,it}]{caption}

\usepackage{setspace}
\usepackage{thm-restate}

\usepackage{epsfig}
\usepackage{amsthm}

\usepackage{mathrsfs}
\usepackage{xspace}
\usepackage{soul}
\usepackage{latexsym}
\usepackage{bbm}

\usepackage{framed}

\usepackage[dvipsnames]{xcolor}
\definecolor{DarkGray}{rgb}{0.66, 0.66, 0.66}
\definecolor{DarkPowderBlue}{rgb}{0.0, 0.2, 0.6}
\definecolor{fluorescentyellow}{rgb}{0.8, 1.0, 0.0}
\definecolor{cerulean}{rgb}{0.0, 0.48, 0.65}
\definecolor{bleudefrance}{rgb}{0.19, 0.55, 0.91}

\usepackage[ruled,vlined,linesnumbered,algonl]{algorithm2e}
\SetEndCharOfAlgoLine{}
\SetKwComment{Comment}{\footnotesize$\triangleright$\ }{}

\SetCommentSty{mycommfont}

\usepackage{thmtools,thm-restate}

\allowdisplaybreaks

\sethlcolor{fluorescentyellow}

\newcommand{\initOneLiners}{%
    \setlength{\itemsep}{0pt}
    \setlength{\parsep }{0pt}
    \setlength{\topsep }{0pt}
}

\newtheorem{assumption}[theorem]{Assumption}

\theoremstyle{definition}
\newtheorem{defn}[theorem]{Definition}

\makeatletter 
\renewcommand{\theinvariant}{(I\@arabic\c@invariant)}
\makeatother

\newcommand{\ent}{\mathrm{H}}
\newcommand{\topk}{{\sc Top}-$k$\xspace}
\newcommand{\lru}{{\sc LRU}\xspace}

\newcommand{\hsub}{\mathscr{H}\xspace}
\newcommand{\lfu}{{\sc LFU}\xspace}

\newcommand{\page}{{\textsf{Paging}}\xspace}

\newcommand{\eps}{\varepsilon}

\newcommand{\opt}{{\textsf{opt}}\xspace}

\newcommand{\EE}{\mathbb{E}}

\newcommand{\TT}{\mathbb{T}}

\newcommand{\poly}{\operatorname{poly}}

\newcommand{\cD}{{\mathcal{D}}}

\newcommand{\E}{\mathcal{E}}

\newcommand{\nf}{\nicefrac}

\newcommand{\tailp}{p(\bT)}

\newcommand{\bH}{{\mathbb H}}
\newcommand{\bT}{{\mathbb T}}

\newcommand{\junk}[1]{}
\newcommand{\eat}[1]{}

\renewcommand{\setminus}{\smallsetminus}

\newif\ifhideproofs

\ifhideproofs
\usepackage{environ}
\NewEnviron{hide}{}

\fi

\usepackage{xr}
\usepackage{mathtools}

\DeclareMathOperator*{\probability}{\Pr}
\newcommand\prob[1]{\probability\left[ #1 \right]}

\DeclareMathOperator*{\expectation}{\mathbb{E}}
\newcommand{\expect}{\expectation\expectarg}
\DeclarePairedDelimiterX{\expectarg}[1]{[}{]}{%
	\ifnum\currentgrouptype=16 \else\begingroup\fi
	\activatebar#1
	\ifnum\currentgrouptype=16 \else\endgroup\fi
}

\DeclarePairedDelimiterX{\nicesetarg}[1]{\{}{\}}{%
	\ifnum\currentgrouptype=16 \else\begingroup\fi
	\activatebar#1
	\ifnum\currentgrouptype=16 \else\endgroup\fi
}

\newcommand{\innermid}{\nonscript\;\delimsize\vert\nonscript\;}
\newcommand{\activatebar}{%
	\begingroup\lccode`\~=`\|
	\lowercase{\endgroup\let~}\innermid 
	\mathcode`|=\string"8000
}

\title{Stochastic Caching via Subset Entropy}

\author{Ravi Kumar}{Google Research, Mountain View CA, USA \and \url{https://sites.google.com/site/ravik53/}}{ravi.k53@gmail.com}{https://orcid.org/0000-0002-2203-2586}{}

\author{Roie Levin}{Rutgers University, Piscataway NJ, USA \and \url{https://roielevin.com/}}{roie.levin@rutgers.edu}{https://orcid.org/0000-0003-2907-7186}{}

\author{Joseph (Seffi) Naor}{Technion -- Israel Institute of Technology, Haifa, Israel \and \url{https://naor.cswp.cs.technion.ac.il/}}{naor@cs.technion.ac.il}{https://orcid.org/0009-0009-4234-9466}{Supported in part by ISF grant 3001/24 and United States – Israel
BSF grant 2022418.}

\author{Debmalya Panigrahi}{Duke University, Durham  NC, USA \and \url{https://www.debmalyapanigrahi.org/}}{debmalya@cs.duke.edu}{https://orcid.org/0000-0003-1799-6660}{Supported in part by NSF grants CCF-2329230 and CCF-1955703.}

\category{APPROX}

\authorrunning{R.\ Kumar, R.\ Levin, J.\ Naor, and D.\ Panigrahi} 

\Copyright{Ravi Kumar, Roie Levin, Joseph Naor, and Debmalya Panigrahi} 



\category{} 

\relatedversion{} 

\nolinenumbers 

\EventLongTitle{29th International Conference on Approximation Algorithms for Combinatorial Optimization Problems (APPROX 2026)}
\EventShortTitle{APPROX 2026}
\EventAcronym{APPROX}
\EventYear{2026}
\ArticleNo{1}

\begin{document}

\maketitle

\begin{abstract}
    A classic approach to beyond worst-case algorithm design is to impose stochastic assumptions on the input. However, a limiting feature of such stochastic analyses is that, by the min-max principle, performance on worst-case distributions mirrors that of randomized algorithms on worst-case inputs. In other words, the same  shortcoming of worst-case analysis---its inability to distinguish between ``easy'' and ``hard' instances---reappears as an inability to distinguish between ``easy'' and ``hard'' distributions. This raises a natural question: {\em Can we characterize ``easy'' input distributions with useful beyond worst-case bounds?}
    
    A canonical example is the {\em stochastic caching} problem (Aho, Denning, and Ullman, 1971).  When the page requests are drawn i.i.d. from the uniform distribution, the best achievable competitive ratio is $O(\log k)$, matching the performance of the best randomized algorithm on worst-case instances (Fiat, Karp, Luby, McGeoch, Sleator, and Young, 1991).  However, when the input distribution has less entropy, intuition suggests that we should be able to do better by exploiting the information provided by the distribution. In this paper, we formalize this intuition by defining a new information-theoretic parameter of probability distributions called {\em subset entropy}. We then use this parameter to give a fine-grained characterization of the competitive ratio of stochastic caching, including a new analysis for the well-known {\rm LRU} algorithm on stochastic inputs.

    While our technical results are for the caching problem, we believe the broader principle---parameterizing algorithmic performance by an entropy measure of the input---is of independent interest and might apply to other online and stochastic optimization problems.  Indeed, for many other fundamental problems such as (comparison-based) sorting, online matching, online load balancing, etc., etc., the hardest stochastic instances involve high-entropy distributions. We hope our work is a step toward a broader theory of fine-grained algorithmic performance for this class of problems.  
\end{abstract}

\pagenumbering{gobble}

\clearpage

\pagenumbering{arabic}

\section{Introduction}

Stochastic optimization is a well-known and extensively-studied framework for beyond worst-case analysis of algorithmic problems. The basic premise is to make an assumption about the input being drawn from a probability distribution, rather than being generated adversarially, thereby enabling algorithms that exploit structural properties of the distribution.  However, a common shortfall of this strategy is that algorithmic performance on the worst-case input distribution often mirrors that of the best randomized algorithm for worst-case inputs. This equivalence, a consequence of the min-max principle (Yao~\cite{Yao77}), is often useful for proving randomized lower bounds for algorithms but, conversely, also lower bounds the performance of the best stochastic algorithms. This is rather limiting because it brings back the main deficiency of worst-case analysis---its inability to separate ``easy'' from ``hard'' instances---as being unable to differentiate between ``easy'' and ``hard'' distributions. This raises the natural question: {\em can we obtain a finer-grained characterization of input distributions that yields beyond worst-case bounds for ``easy'' distributions?}

\par\noindent{\bf The Caching Problem.} A case in point is the classical online \emph{caching} (or paging) problem. In this problem, a sequence of page requests arrives online from a universe of $n$ pages. The goal is to serve each request using a cache that can only hold $k$ ($\ll n$) pages at a time while minimizing the total number of page swaps.
Caching is used in almost every extant computer system. It is also a focal point for modern developments in online algorithms such as the
online primal-dual method, projections, and mirror descent~\cite{BNsurvey, BCN14, BCLLM18}.
The problem is considered well-understood in online adversarial models: 
in their seminal work, Sleator and Tarjan \cite{ST85}
showed that the \emph{least-recently used} (\lru) policy is $k$-competitive, and that any deterministic online algorithm is at least $k$-competitive. 
When randomization is allowed, Fiat et al. \cite{F+} showed an elegant and improved $O(\log k)$-competitive algorithm. Unfortunately, this bound is asymptotically tight for randomized algorithms even when page requests are sampled i.i.d.\ from a known distribution, since no algorithm can do better on the uniform distribution over $n=k+1$ pages \cite{F+}.

Yet, our understanding of the caching problem is wanting even in the simplest ``beyond worst-case'' setting, where each request is drawn independently from a fixed distribution. As noted, the uniform distribution \emph{is} the worst-case distribution---but this is also the least compelling case for caching. If requests are uniform, there is no reason for a caching algorithm to prefer one page over another, and any caching algorithm is as good as any other. In contrast, if page requests are skewed toward a set of fewer than $k$ pages, smart caching can be effective, as one can hope to identify the frequently requested pages and prioritize evicting others. 

This leads to our main questions:  
\textsl{Which distributions admit better-than-worst-case competitive ratios? And which online caching algorithms attain tighter bounds for distributions that are caching-friendly?}    

\subsection{Our Results}

We study the {\em stochastic caching} problem, where we assume the input sequence of requests is drawn independently from a fixed distribution. The {\em competitive ratio} of a stochastic caching algorithm is the ratio of the expected cost of the algorithm and the expected cost of an optimal offline solution, where both expectations are taken over the input distribution (and the algorithm, if randomized). 

Our first contribution is a new measure called {\em subset entropy} to quantify the competitive ratio of a stochastic caching instance. A reasonable first guess for such a  measure would be the \emph{entropy} of the input distribution; indeed, this would align in the case of the  uniform distribution over $n=k+1$ pages, where the entropy is $\Theta(\log k)$. 
However, if the distribution places $1-\eps$ mass on one page, and $\eps/(n-1)$ on each of the remaining pages, then the entropy is near zero, but this is essentially the same stochastic caching instance as before, but with a cache of size $k-1$, and with the $n-1$ pages of low probability in the original instance.

This motivates our definition of subset entropy for a distribution $\cD$: it is the maximum, taken over all subsets in the support of a distribution $\cD$, of the entropy of $\cD$ conditioned to the subset. In the case of stochastic caching, we only need to consider subsets of size at most $k$ in the support; we call this {\em $k$-subset entropy} and denote it $\hsub(\cD, k)$, or $\hsub$ for short when $\cD$ and $k$ are clear from context. In \Cref{known distribution}, we show how to achieve competitive ratio $O(\hsub)$ for every input distribution $\cD$.\footnote{
There do exist distributions for which $\hsub$ is large and yet the distribution admits smaller competitive ratio. E.g., the uniform distribution with $n = 10k$ pages has $\hsub = O(\log k)$ but any algorithm is constant-competitive. To our knowledge this is true of most parameterized results, since it is rare to identify a parameter that guarantees \emph{both} upper and lower bounds.}
\begin{theorem}[Informal]
\label{thm:main1}
    There is a caching algorithm that for a given distribution $\cD$ and a sequence of page requests drawn i.i.d.\ from $\cD$, achieves a competitive ratio of $O(\hsub)$.
\end{theorem}

Note that $\hsub$ is at most the worst case bound $O(\log k)$, but it is also much sharper for many distributions. In fact, we show in \Cref{sec:prelim} that the only distributions for which $\hsub = \Theta(\log k)$ are those that are nearly uniform, namely those for which some subset of $\Omega(\poly(k))$ pages have probabilities within a constant factor of each other.
To our knowledge, the notion of subset entropy is novel and would be interesting to explore in other contexts as well. For example, does it encapsulate interesting information-theoretic properties of a distribution? We show one property in \Cref{sec:prelim}: informally, we show that it captures the cardinality of the largest (approximately) uniform distribution embedded in the overall distribution.  A different question: is $\hsub$ (exactly) computable efficiently?  
We answer this in the affirmative in \Cref{sec:subset-algo},  
given explicit access to the distribution.  We hope this new parameter will be studied from other angles in the future.

The algorithm we use to establish \Cref{thm:main1} is natural: since $\cD$ is known, we keep the top $k$-1 pages by probability in cache permanently and use the last cache slot to serve all other requests;\footnote{For ease of analysis, we define the \topk algorithm slightly differently from the optimal online strategy; see \Cref{known distribution}.} we call this algorithm \topk. Indeed, in an important work from 1971, Aho, Delling, and Ullman~\cite{AhoDU71} showed that this is the {\em optimal} online caching strategy in the i.i.d.\ setting.\footnote{In the stochastic setting, an optimal online solution is well-defined. It can be obtained from a Markov decision process via a dynamic program which, in general, can have an exponential-sized state space.} Though the algorithm is unsurprising, the analysis is new. We show that \topk is $O(\hsub)$-competitive against the  \emph{offline} optimum. This quantifies the precise role of randomness in the input sequence on the competitive ratio, namely the difference between knowing the input \emph{distribution} versus the actual \emph{realization} of the input sequence. 

The previous \topk algorithm requires knowing the input distribution (we call this the {\em known} i.i.d.\ setting).  Now, we consider the {\em unknown} i.i.d.\ setting where the caching algorithm no longer has access to the input distribution. Intuitively, we would still like to keep the top $k$ pages by probability in the cache. Ordering pages by decreasing probability is equivalent to ordering them by increasing expected length of the interval between consecutive requests. Since we do not have access to the expected interval lengths, we use as a surrogate the length of the interval since the last request for every page---i.e., keep in cache the $k$ most recently requested pages.  This is none other than the well-known \lru algorithm, and our second result shows that it is $O(\hsub^3)$-competitive even in the unknown i.i.d.\ setting.
\begin{theorem}
\label{thm:main2}
    When given a sequence of page requests drawn i.i.d.\ from an unknown distribution $\cD$, the \lru caching algorithm achieves a competitive ratio of $O(\hsub^3)$.
\end{theorem}
As intuition would predict, \lru approximately learns which pages in cache are most likely to be requested furthest in future and evicts these first. How well it does so depends on how predictable is this information, which in turn is captured by the subset entropy of $\cD$. Since $\hsub = O(\log k)$, this also shows that \lru is $O(\log^3 k)$-competitive for unknown i.i.d.\ inputs. To our knowledge, this is the first analysis of \lru in a stochastic setting that yields an $o(k)$-competitive factor.

A more precise (although also not as lightweight) alternative to \lru is the \emph{least-frequently used} (\lfu) algorithm that always evicts the page with the least number of requests so far. Our third result shows that LFU is $O(\hsub)$-competitive for i.i.d.\ inputs. This implies that up to constant factors, the known and unknown i.i.d.\ settings have the same competitive ratio. 

\begin{restatable}{theorem}{lfuthm}
\label{thm:main3}
    When given a sequence of page requests drawn i.i.d.\ from an unknown distribution $\cD$, the LFU caching algorithm achieves a competitive ratio of $O(\hsub)$.    
\end{restatable}

This theorem is not surprising in an asymptotic sense, since LFU eventually converges to the \topk algorithm. But a priori, the rate of convergence could be arbitrarily slow depending on the difference between the probability values of pages requests. We show that the additive loss of the LFU algorithm vis-\`a-vis the \topk algorithm is  $O(k \hsub)$, and is therefore absorbed in the multiplicative factor in \Cref{thm:main3}.\footnote{$k$ is required for all algorithms as an initialization cost from an empty cache.} In other words, LFU matches the performance of \topk \emph{even when it has not finished learning the input distribution}. 
We view this as an interesting insight into learning-augmented algorithms: it demonstrates the difference between statistically learning a distribution and learning it well enough to devise an effective algorithmic strategy.

\subsection{Techniques and Overview}

We start in \Cref{sec:prelim} with our characterization of $k$-subset entropy in terms of the largest (up to size $k$) approximately uniform distribution embedded in $\cD$.  Indeed, by definition, any conditional distribution in $\cD$ gives a lower bound on subset entropy.  To show the converse, we use a factor-revealing LP to bound the maximum value of subset entropy when every approximately uniform distribution is size-constrained.

In \Cref{known distribution}, we show that the \topk algorithm has a competitive ratio of $O(\hsub)$, thereby establishing \Cref{thm:main1}. The $k$ pages with the largest values of $p_i$ are permanently in the cache in the \topk algorithm; we call them the \emph{head} of the distribution and call the remaining pages the \emph{tail}.  The expected cost of the \topk algorithm is immediate: it is the total probability in the tail, denoted $\tailp$. But how do we lower bound the expected cost of an offline optimal solution? The usual technique is to use dual fitting, i.e., write primal and dual LPs and then construct a feasible dual solution of sufficiently large value. Here, since the instance itself is random, so are the coefficients of the primal and dual LPs, and what we really have is a \emph{distribution} over LPs. How do we argue about the expected value of the optimal solution of such a distribution? (Contrast this to randomized dual fitting~\cite{DevanurJK13} where the LPs are fixed but the primal and dual solutions are random.)

We start by writing the ``expected'' primal-dual LPs whose coefficients are the expected values of the (random) coefficients over the distributions of the primal and dual LPs. For this LP, we construct a feasible dual solution of value $\tailp$ that satisfies these ``expected constraints''. Interestingly, this dual solution to the ``expected'' LP turns out to be quite different from those used in standard primal-dual analysis of caching algorithms (e.g., in \cite{BBN08}). However, we want the dual solution to be feasible for the actual realized constraints, not the expected constraints. Constraints for which the realized dual coefficients are ``close'' to their expected values (in a specific sense that we define later) are already approximately satisfied by the expected dual. The main challenge is to handle the ``tail constraints'', namely those that realize very differently from their expected form. These realizations can happen in many different ways and tracking their feasibility precisely is onerous. Instead, we carefully zero out dual variables to make sure that the tail constraints are also satisfied. This creates a distribution over dual solutions, since the set of tail constraints, and therefore, zeroed out variables, depends on the particular realization. Nevertheless, crucially, for each realized dual instance, the corresponding dual is \emph{always} feasible. The key is to show that since each constraint is a tail constraint with small probability, we do not lose too much in the objective by zeroing out dual variables. In fact, we precisely lose a factor of $O(\hsub)$. Hence this process gives a randomly defined dual solution that is deterministically feasible for the corresponding random instance, and the expected value of the dual objective over these random solutions is $\Omega\left(\tailp / \hsub\right)$. 

In \Cref{sec:unknown2}, we study the \lru algorithm and establish \Cref{thm:main2}. Recall that \lru  is an empirical implementation of the \topk algorithm which evicts pages based on the \emph{expected} length time between subsequent requests. Thus it is tempting to imagine that the ideal scenario for \lru is when the empirical length of a request interval exactly matches its expected value for every page. Unfortunately, this intuition is misleading. Consider a uniform input distribution $\cD$ over $n = k+1$ pages. If an instantiation creates a round-robin request sequence, then the empirical length of a request interval is exactly $k+1$, i.e., its expected value. However, \lru is $k$-competitive for this input sequence! Therefore, we need to use the symmetry-breaking properties of a random sequence in our analysis. 
Suppose a page $i$ is requested at time $t$. If page $i$ is not in the cache at time $t$, then by the properties of \lru, page $i$ is not among the last $k$ distinct pages requested before time $t$. What is the probability of this event? 
By reversing time, this is the same as the probability that page $i$ is not going to be among the next $k$ distinct pages to be requested after time $t$. 
Now, define a random variable for every page denoting the time till the first request of the page. Although these times are dependent, to first approximation, we can imagine that each page $i$ is running an independent Poisson process with rate $p_i$, i.e., we can approximate the length of the request intervals by exponential variables. Then, we use the properties of Poisson processes to bound the probability that at least $k$ other processes fired before process $i$, eventually proving \Cref{thm:main2}. 

In \Cref{sec:lfu}, we study the LFU algorithm and prove \Cref{thm:main3}. We note that LFU might suffer more cache misses than \topk because of inversions in the ranked order of pages by empirical frequency counts vis-\`a-vis their true probabilities. We bound the effect of these inversions in two parts. First, we show that inversions due to pairs that have probabilities within (say) a factor of 2 only incur an expected cost $O(\tailp)$. It is crucial that these inversions do not cost us in the additive term because such inversions might persist for an arbitrarily long time (e.g., if two pages differ in probability by $\eps$). It remains to bound the cost due to inversions where one page is twice as likely as the other. We show that after $O(\hsub)$ requests for a page, its count will forever be higher (with good probability) than the count of any page with less than half its probability. This incurs an additive loss of $O(k\hsub)$ on the pages in the head, which is ultimately absorbed in the multiplicative factor.

\par\noindent{\bf The Bigger Picture.} 
Having discussed our technical results for stochastic caching, we now return to the broader perspective: obtaining finer-grained analyses for stochastic optimization problems. We argue that identifying a distributional quantity, such as entropy, and analyzing algorithms with respect to this quantity, offers a promising direction for beyond worst-case analysis in stochastic settings.  Indeed, for many fundamental problems---such as comparison-based sorting, online matching, online load balancing, and related tasks---the
hardest stochastic instances correspond to high-entropy (often uniform) distributions. Intuitively, this is consistent with the view
that distributions with higher entropy are less informative than those with lower entropy.  We view our work as a first step toward a theory, based in entropy and its variants, that provides a refined analysis of algorithmic performance across stochastic versions of this rich class of problems.

\subsection{Related Work}
In a classic paper, Belady \cite{belady} introduced the optimal ``furthest in the future'' rule for evicting pages from a cache in the offline setting. Soon after, stochastic versions of the online caching problem were considered and adaptations of Belady's rule to online distributional settings were proposed and analyzed~\cite{AhoDU71,FranaszekW74}. Over the next few decades, a rich array of input models were introduced for online caching; some highlights include the work of Lund et al. \cite{LPR99} and Pabbaraju and Vakilian \cite{PV25}, who obtained approximations against the best online strategy even in partial information settings, the Markov paging model of Karlin et al. \cite{KarlinPR00}, where the input sequence is a Markov chain, and the access graph setting of Borodin et al. \cite{BorodinIRS95} (see also \cite[Ch. 5]{BY98}), where the request sequence is generated by a simple walk on an ``access graph'' defined on the pages. Somewhat related to our work is that of Pandurangan and Upfal \cite{PU07}, who studied the relationship between the performance of the best online caching algorithm and the entropy of a request sequence generated by a stochastic process. In particular, they showed that entropy is not sufficient to capture the performance of online caching algorithms.

We note that many further variants of online caching have been proposed, e.g., when pages have non-uniform weights \cite{BBN07, I02, Y16a}, non-uniform sizes \cite{AAK99, BBN08, I97}, 
as well as non-standard caching models such as elastic caches~\cite{GuptaK0P19}, caching with time windows~\cite{GKP20}, and caching
with dynamic weights~\cite{EvenMR18}, interleaved caching~\cite{KPSV20}, and caching with machine learning predictions~\cite{LykourisV18, BansalCKPV22}. In contrast to the i.i.d. setting that we study in this paper, many models try to capture \emph{locality of reference}, i.e., the phenomenon that pages requested recently tend to be re-requested~\cite{KoutsPapa00,Young00,Becchetti04,PanagSouza06}. The validity of locality of reference varies by application domain, e.g., is more applicable for centralized settings like operating systems on computing devices compared to distributed and resource-shared settings like data centers where the requests come from a large number of independent processes. Nevertheless, we view the extensive literature on stochastic paging---including our work---and that on models capturing locality of reference as complementary, since they capture different aspects of a fundamental problem. Indeed, this complementarity can sometimes be leveraged to obtain richer models that try to incorporate both phenomenon, such as Markov paging~\cite{KarlinPR00}. 

In our work, the input model is the classic i.i.d.\ setting, but we are interested in characterizing input distributions that are amenable to caching and obtaining better-than-worst-case bounds for such inputs. In this sense, it aligns with the recent focus on ``beyond worst-case'' algorithm design, and in particular, with the emerging field of learning-augmented algorithms, e.g., \cite{site:ALPS}. The main difference is that while the typical goal in the latter context is to use additional information about the input and automatically adapt the performance bounds to the {\em quality} of this information, our goal is to adapt  performance bounds to the quality of the input distribution itself without using additional information about the input. We point the reader to an excellent survey of ``beyond worst-case'' caching models in \cite[Ch. 24]{R2020} for a detailed discussion of related models.

\section{Problem Definition and Background}
\label{sec:prelim}
\subsection{The Stochastic Caching Problem}

There is a universe of $n$ pages indexed $[n] = \{1, \ldots, n\}$, a cache of size $k \le n$, a request sequence of length $T$ indexed by $[T]$, and a distribution $\cD$ on $[n]$. At each time $t \in [T]$, the algorithm gets a page request $r_t\in [n]$, where $r_t$ is drawn independently from $\cD$. The algorithm maintains at most $k$ pages in the cache at all times with the constraint that page $r_t$ must be in the cache at time $t$. Swapping in/out a page in the cache incurs unit cost; the goal is to minimize the total cost, i.e., the number of page swaps. We call this the {\em stochastic caching} problem.

The distribution $\cD$ is defined by probabilities $p_1, \ldots, p_n$, where $p_i > 0$ is the probability that page $i\in [n]$ is requested at any given time. Clearly, $\sum_{i} p_i = 1$.  We also denote $p(S) = \sum_{i \in S} p_i$ for a subset $S \subseteq [n]$. Wlog, we index pages in decreasing order of probabilities, i.e., $p_1 \ge p_2 \ge \cdots \ge p_n$, breaking ties arbitrarily.  The top $k$ pages are called the \emph{head} of the distribution $\cD$ and denoted $\bH := [k]$; the remaining pages $\bT := [n]\setminus [k]$ comprise the \emph{tail} of the distribution. The total probability mass in the tail is denoted $\tailp := \sum_{i=k+1}^n p_i = 1 - p(\bH)$.  In this paper, we consider two cases: when $\cD$ is \emph{known} and when $\cD$ is \emph{unknown} to the algorithm.

\subsection{Subset Entropy}

For a distribution $\cD$ with support $U$ and for a subset $S \subseteq U$, let $\cD_S$ denote the distribution $\cD$ conditioned on $S$, i.e., $\cD_S(u) = \frac{\cD(u)}{\sum_{u\in S} \cD(u)}$ for $u\in S$.  Let $\ent(\cdot)$ denote the \emph{Shannon entropy} function:
$
    \ent(\cD_S) := - \sum_{u\in S} \cD_S(u) \cdot \log \cD_S(u).
$
(All logarithms throughout the paper are base $2$ unless otherwise stated.)

We define {\em $k$-subset entropy} of a distribution $\cD$ as follows:
\begin{defn}[$k$-Subset Entropy]
For a discrete distribution $\cD: U \rightarrow [0, 1]$ and any positive integer $k\le |U|$, the {\em $k$-subset entropy} of $\cD$ is defined as the following:
\begin{align}
    \hsub(\cD, k) = \max_{S\subseteq U: |S|\le k} \ent(\cD_S).\label{eq:subsetH}
\end{align}
\end{defn}

While $k$-subset entropy suffices in the context of caching, an interesting information-theoretic counterpart is the unrestricted version that maximizes entropy of conditional distributions over {\em all} subsets of $U$. We call this {\em subset entropy}, and it is the special case of $k$-subset entropy when $k = n$.

Next, we prove an important fact about $k$-subset entropy, which we rely on in the rest of the paper. 

\begin{restatable}[$k$-Subset Entropy Lemma]{lemma}{subentlem}
\label{lem:dominate}
    Suppose we have a distribution $\cD: U\rightarrow [0, 1]$. Let $S_\ell = \{u\in U: 1/2^{\ell+1} < \cD(u) \le 1/2^{\ell})\}$ for any $\ell \ge 1$ and let $n_{\max}: = \max_\ell |S_\ell|$. Then, 
    \[
        \hsub(\cD, k) = \Theta(\log (\min \{k, n_{\max} \})).
    \]    
\end{restatable}

One direction is immediate: clearly $\hsub(\cD, k) \geq \log (\min(k,n_{\max}))$ since the set $S$ in the definition \eqref{eq:subsetH} could be chosen as (a subset of size $\leq k$ of) the $S_\ell$ for which $n_{\max} = |S_\ell|$, and by the definition of Shannon entropy, $H(D_{S_\ell}) = \Theta(\log |S_{\ell}|)$. Surprisingly, the converse holds: if $\hsub(\cD, k) = h$, then there is a set $S\subseteq U, |S| \ge 2^{\Theta(h)}$, over which the conditional distribution $\cD_S$ is uniform up to constants. Intuitively, this means that the $k$-subset entropy of a distribution is dictated by the largest (up to $k$) approximately uniform distribution embedded in it. The proof (\Cref{sec:prelim_proofs}) is via dual fitting of a factor revealing LP: we show that even if we ignore the original distribution $\cD$ and choose the probabilities of the pages in the maximizing set $S$ completely freely, if we are constrained to choosing at most $n_{\max}$ pages within a factor $2$ of each other, we cannot give $S$ a conditional entropy greater than $O(\log n_{\max})$.

In \Cref{sec:subset-algo}, we also show that the maximizing set in the Definition \eqref{eq:subsetH} is always a subinterval of $U$, when arranged in sorted order of 
 probabilities. This gives an efficient algorithm for computing $k$-subset entropy \emph{exactly}. Nevertheless, the approximate characterization above is convenient to work with in our applications.

We use henceforth the shorthand $\hsub$ to denote the $k$-subset entropy $\hsub(\cD, k)$, since it is implicit that the distribution $\cD$ is the input distribution and the parameter $k$ is the cache size.

\section{Stochastic Caching for a Known Distribution}
\label{known distribution}
We begin with the known distribution setting. Recall our simple algorithm \topk: keep the pages of the head $\bH$ in cache; when a page from the tail $\bT$ is requested, swap out an arbitrary page in cache for the requested page, then immediately swap the same page back. The main result of this section, which establishes \Cref{thm:main1}, is the following:

\begin{theorem}
\label{thm:ratio}
    The competitive ratio of the algorithm \topk for the stochastic caching problem for input distribution $\cD$ is $O(\hsub)$.
\end{theorem}
To establish \Cref{thm:ratio}, we first bound the expected cost of the \topk algorithm.
\begin{lemma}
\label{lem:topk}
    The expected cost of the \topk algorithm for each request is $2\tailp$.
\end{lemma}
\begin{proof}
    For each request, the probability that the requested page is not in the cache is exactly $\tailp$, in which case the algorithm performs two evictions.
\end{proof}

The harder challenge is to lower bound the expected cost of an offline optimal solution. To this end, we fit a dual to the LP relaxation of the knapsack-cover (or KC) integer programming formulation for caching (see, e.g., \cite{BBN08})\footnote{Note that our choice of the KC inequality strengthened formulation is for convenience; it is not strictly necessary because the standard LP formulation is already integral.}. As noted in the introduction, while focusing on LP duality is a reasonable instinct, the main challenge is that the LP is defined via a distribution over coefficients, i.e., both the primal and dual LPs are random. We have to produce a feasible dual solution for \emph{every} realization of the random instance, and then bound the expected objective value across the realizations.

Fortunately, since we are in an i.i.d.\ setting, the distribution over dual constraints that apply at any given time $t$ is time-invariant, i.e., is independent of $t$. We view this distribution over dual constraints in two parts---the set of constraints that apply within a high probability interval, and the remaining constraints that are in the tail of the distribution. For the former, we set dual variables corresponding to the ``expected'' constraints and these are approximately feasible. For the latter, there are a large number of possibilities and it is onerous to track feasibility across all these realizations. Instead, we carefully zero out a subset of dual variables to ensure feasibility for the constraints in the tail of the distribution. Of course, the set of dual constraints that fall in the tail region is random, and therefore the set of zeroed out variables is also random. Nevertheless, we show that since a constraint is in the tail with low probability (by standard concentration inequalities), this last step of zeroing out variables to ensure feasibility has limited impact on the expected value of the dual objective. 
\begin{restatable}{lemma}{optlb}
    \label{lem:opt_lb}
    The expected cost of an offline optimal solution is $\Omega\left(\frac{\tailp}{\hsub}\right)$.
\end{restatable}

\begin{proof}[Proof Sketch]
Let $x(i, r)$ be the indicator variable for the event that page $i$ is evicted in the time interval $I(i,r) := [t(i,r)+1,t(i,r+1)-1]$, where $t(i,r)$ is the time of the $r$th request to page $i$.  Let $\rho(i, t)$ be the number of requests to page $i$ up to and including time $t$.  Consider the formulation: 
\[
{\small
	\boxed{\begin{array}{cc}
			 {\displaystyle  \min \sum_{i =1}^n \sum_{r=1}^{\rho(i,T)} x(i,r)} & \text{subject to }\\
			\\
			 {\displaystyle \sum_{i \in S\setminus\{r_t\}} x(i, \rho(i,t)) \geq |S| - k} & \begin{array}{c}
			     \forall t \in [T] \\
			     \forall S \subseteq [n]
			\end{array}\\
			\\
   x(i,r) \geq 0 & \begin{array}{c}
			     \forall i \in [n] \\
        \forall r \in [\rho(i,T)]
			\end{array}
	\end{array}}\qquad\boxed{\begin{array}{cc}
	{\displaystyle \max\sum_{t=1}^T\sum_{S \subseteq [n]} (|S| - k) \cdot y_S^t} & \text{subject to}\\
	\\
        {\displaystyle \sum_{t \in I(i,r)} \sum_{S: \, i \in S} y_S^t\leq 1} & \begin{array}{c}
             \forall i \in [n]  \\
             \forall r\in [\rho(i,T)] 
        \end{array}\\
	\\
        y_S^t \geq 0 & \begin{array}{c}
            \forall t \in [T] \\
             \forall S \subseteq[n]
        \end{array}
\end{array}}
}
\]

This formulation is standard, but our dual fitting is not (see \cite{BBN08} for the usual way the online dual is fit using multiplicative weights).
We start by finding a stationary solution to the \emph{expected} dual, i.e., find a vector $y$ that is the same for each $t \in [T]$, such that the expected dual constraint is satisfied. 

Intuitively, we should expect every time $t$ to lie in an interval of length $1/p_i$ between subsequent requests to page $i$, and hence these expected constraints would ask that there be less than $1$ total $y$ mass on page subsets containing $i$ on these intervals, i.e., that there be less than $p_i$ mass on these subsets per time step $t$.  Define 
\[y_{S}^t = \begin{cases}
    p_{i} - p_{i+1} & \text{if } S = [i] \text{ and } i \in \bT \setminus \{n\}, \\
    p_{n} & \text{if } S = [n], \\
    0 & \text{otherwise}.
\end{cases}\]
  Then, by design, for all $i \in \bT$ we have
    $\sum_{S: \, i \in S} y_S^{t} = \sum_{j=i}^{n-1} \left(p_{j} - p_{j+1}\right) + p_{n} = p_{i}$,
    as desired.
    For the remaining $i \in \bH$, we have
    $\sum_{S: \, i \in S} y_S^{t} = \sum_{j=k+1}^{n-1} \left(p_j - p_{j+1}\right) +  p_n = p_{k+1} \leq p_i$,
    Furthermore, the dual objective (per time step) is $\sum_{S} (|S| - k) \cdot y_S^t = n \cdot p_{n} + \sum_{j=1}^{n} j \cdot (p_j  - p_{j+1}) = \sum_{i \in \bT} p_{i}$.
    
However, this is too good to be true!  If this were a legal dual solution, weak duality would imply that \topk in fact achieves competitive ratio $O(1)$, which we know is impossible, as demonstrated by the uniform distribution case. Indeed, for the dual to witness a true lower bound on the primal program, it must be feasible with probability $1$, and our stationary solution, as described, might violate dual constraints corresponding to intervals that happen to be longer than their expectation.

To construct our actual dual, we start with the stationary $y$ described above. We then show how to carefully drop some coordinates to account for the random fluctuation in interval lengths. Since the dual is a packing problem, we can always make a solution feasible this way. The crux is to do this while approximately preserving the objective value. In \Cref{sec:topk_proofs}, we give the details of our modified dual construction and show that we only lose an $O(\hsub)$ factor in the objective. 
\end{proof}

\section{The LRU Algorithm}\label{sec:unknown2}In this section, we analyze the popular \emph{least recently used} (\lru) algorithm for stochastic caching. Note that \lru does not use distributional information, and hence can be used in the unknown distribution setting. Our main result is:

\begin{theorem}
	\label{thm:lru}
	The \lru algorithm has a competitive ratio $O(\hsub^3)$ for the stochastic caching problem.
\end{theorem}
We now provide the main ideas needed for proving this and defer the technical claims to \Cref{sec:lru-proofs}.  We saw in \Cref{known distribution} that the 
expected cost of an offline optimal solution is $\Omega\left(\frac{\tailp}{\hsub}\right)$.  
Hence, it suffices to bound the expected cost of \lru as $O(\hsub^2 \cdot \tailp)$.
Note that \lru has a cache miss when the requested page is not among the last $k$ distinct pages requested. Thus, for any page $i$, the probability that $i$ incurs a cache miss (at any point of time) is the same as that of page $i$ not being among the first $k$ distinct pages requested in a sequence of pages drawn from $\cD$. We denote this probability $q_i$ for page $i$. With this notation, the expected cost of the \lru algorithm for any request is $\sum_{i \in [n]} p_i q_i$. Our goal from here on is to show 
\begin{equation}\label{eq:main}
    \sum_{i \in[n]} p_i q_i \le O(\hsub^2\cdot \tailp).
\end{equation}

Before we begin, note that when $\tailp > 1/2$, any algorithm, even one that incurs a cache miss at every time step, has cost $O(\tailp)$. Thus, we can assume without loss of generality that: 
\begin{assumption}
	\label{assmp:tail}
	The tail has total probability $\tailp \leq 1/2$.
\end{assumption}

Note that the contribution of the tail to the LHS of \eqref{eq:main} is easy to handle, since the $q_i$'s are probabilities.
\begin{observation}[Bounding $\bT$]
	\label{obs:tail}
	The contribution of $\bT$ is bounded as $\sum_{i \in \bT} p_i q_i \leq \tailp$.
\end{observation}

This leaves us with the task of bounding the contribution of the head $\bH$ to the LHS of \eqref{eq:main}.

\subsection{Probability Levels}

We first group the pages in the head $\bH$ into doubling intervals by probability values $p_i$. In particular, for $\ell\ge 0$, we use $L_\ell$ to denote the set of all pages $i$ with $p_i\in [\nf{1}{2^{\ell+1}}, \nf{1}{2^{\ell}})$; we say that the pages in $p_\ell$ are at level $\ell$ and denote $|L_\ell| = n_\ell$. 
Furthermore, we say that $\ell$ is a {\em dominant level} if $\sum_{i\in L_\ell} p_i > \nf 12$; else, $\ell$ is a {\em non-dominant level}.  Note that there is either no dominant level or a single unique dominant level.

Let $X_\ell$ be the random variable denoting the number of pages in $L_\ell$ that do not appear among the first $k$ distinct pages in any particular sequence of page requests. Let $M_\ell$ be the multiset of page requests up to and including the first time all the pages in $L_\ell$ are requested. We denote $Y_\ell = |M_\ell|$ and $Z_\ell$ to be the number of page requests from $M_\ell$ to pages in $\bT$.

We need certain structural properties of these random variables.

\begin{restatable}{lemma}{xyzstructure}
    \label{lem:prop}
    The following three properties hold for $X_\ell$, $Y_\ell$, and $Z_\ell$:
    \begin{enumerate}[(i)]
        \item For any level $\ell$ with $n_\ell > 0$,
        \begin{equation}\label{eq:indicator}
            \sum_{i\in L_\ell} p_i q_i
           \le \frac{1}{2^\ell}\cdot \sum_{i\in L_\ell} q_i
            = \frac{1}{2^\ell}\cdot \EE[X_\ell].
        \end{equation}
        \item Let $\ell$ be a level with $n_\ell > 0$, and
        let $H_{n_\ell}$ denote the $n_\ell$th Harmonic number. Then,
        \begin{align}
\EE[Y_\ell] \le 2^{\ell+1}\cdot H_{n_\ell}. \label{eq:Yell_bound}
\end{align}
\item   Suppose $\ell$ is a non-dominant level. Then, we have 
    \begin{numcases}
    {\Pr[Z_\ell \ge t] \le  } 
            2\cdot \tailp \cdot 2^{\ell+1}\cdot \hsub & for $t = 1$,\label{eq:tail-alt-1}
            \\
            \left(2e^2\cdot \tailp \cdot 2^{\ell+1}\cdot \hsub\right)^t & for $t \ge 2$.\label{eq:tail-alt}
        \end{numcases}    
\end{enumerate}
\end{restatable}

Next, using these properties, we bound the total contribution of any single level $L_\ell$ to $\sum_i p_i q_i$ (i.e., to the LHS in \eqref{eq:main}):
\begin{claim}\label{cl:single-level}
    For any level $\ell$, we have $~\displaystyle\sum_{i \in L_\ell} p_i q_i \leq 4 \hsub \cdot \tailp$.
\end{claim}
\begin{proof}
 By \eqref{eq:indicator},
    $\sum_{i\in L_\ell} p_i q_i
        \le \frac{\EE[X_{\ell}]}{2^{\ell}}
         \le \frac{\EE[Z_{\ell}]}{2^{\ell}}$ since $X_\ell \le Z_\ell$ for any level $\ell$.
    
    So, it suffices to bound $\EE[Z_{\ell}]$. Recall that $M_\ell$ denotes the sequence of requests until all pages in $L_\ell$ are requested. Color every request in $M_\ell$ red, blue, or green depending on whether it is for a page in $\TT$, for a page in $\bH$ but not in $L_\ell$, and for a page in $L_\ell$, respectively. Variable $Z_\ell$ counts the number of red requests in $M_\ell$. 
    
    Let $N$ denote the number of green requests in $M_\ell$ and let $\Gamma_i$, for $i\ge 1$, be the number of red requests between the $(i-1)$st and $i$th green request. Note that the random variables $\Gamma_i$ are independent and identically distributed; let $\Gamma$ denote the common distribution. Furthermore, $N$ is independent of the sequence $\Gamma_1, \Gamma_2, \ldots$. Therefore, by Wald's identity we have
    \begin{equation}\label{eq:zell}
        \EE[Z_\ell] = \sum_i \EE[\Gamma_i] = \EE[\Gamma] \cdot \EE[N]. 
    \end{equation}
    We bound the two terms in the RHS of \eqref{eq:zell} separately. For $\EE[N]$, we use standard coupon collector bounds since the probabilities of pages in level $\ell$ are uniform up to a factor of $2$:
    $\EE[N] \le 2n_{\ell} H_{n_{\ell}} \le 2n_{\ell} \cdot \hsub$.        

    To bound $\EE[\Gamma]$, note that $\Gamma$ denotes the number of red requests before the first green request. Denote $p(\ell) := \sum_{i\in L_{\ell}} p_i$ as the probability of a green request and recall that $p(\TT)$ is the probability of a red request. Then,
    $\EE[\Gamma] =
        \left(\frac{p(\ell)}{p(\ell)+p(\TT)}\right)^{-1} - 1 
        = \frac{\tailp}{p(\ell)} \leq \frac{\tailp}{n_{\ell}\cdot 2^{-(\ell+1)}}$.
   
     Therefore, by \eqref{eq:zell},
    $\EE[Z_{\ell}] = \EE[\Gamma] \cdot \EE[N] < \frac{\tailp}{n_{\ell}\cdot 2^{-(\ell+1)}} \cdot 2n_{\ell} \cdot \hsub
        = 4\hsub\cdot \tailp\cdot 2^{\ell}$,
thus yielding the claim. \qedhere
\end{proof}

The above bound is sufficiently tight for levels $\ell$ where the probabilities $p_i$ are comparable to (or smaller than) the value of the tail probability $\tailp$. However, in higher levels with larger probabilities, the above bound is quite loose and we need a stronger guarantee. 

Let $r_1$ be the largest index $i$ such that $p_i \geq 4e^2\cdot  \hsub \cdot \tailp$ and let $\ell_1$ be the level that contains page $p_{r_1}$.
Roughly speaking, our tighter bound for levels $\ell < \ell_1$ geometrically shrinks as we go to levels with larger probabilities; this will later allow us to add these bounds across levels in a lossless manner. 

\begin{restatable}{claim}{singlelevelhone}
    \label{cl:single-level-h1}
    For any non-dominant level $\ell < \ell_1$, we have
    $
        \sum_{i\in L_\ell} p_i q_i 
        \le \frac{80 \hsub^2\cdot \tailp}{2^{\ell_1-\ell}}.$
    
\end{restatable}

\begin{proof}[Proof Sketch]
    Since $X_\ell \le Z_\ell$, we get from \Cref{lem:prop}{\em (i)} that
\begin{equation}\label{eq:h1-main}
    \sum_{i\in L_\ell} p_i q_i \le \frac{1}{2^\ell} \cdot \EE[X_\ell] \le
    \Pr[X_\ell = 1 \mid Z_\ell = 1]\cdot \Pr[Z_\ell \ge 1] + \sum_{t\ge 2} \Pr[Z_\ell \ge t],
\end{equation}
where we split between $X_\ell = 1$ and $X_\ell \ge 2$, and replaced $X_\ell$ with $Z_\ell$. 
We bound $\sum_{t\ge 2} \Pr[Z_\ell \ge t]$ and $\Pr[Z_\ell \ge 1]$ using \Cref{lem:prop} {\em (iii)}. The interesting part is to bound $\Pr[X_\ell = 1 \mid Z_\ell = 1]$. 

Note that since $Z_\ell = 1$, there is only one request from $\bT$ before all pages in level $\ell$ are requested. Therefore, a necessary event for $X_\ell = 1 \wedge Z_\ell = 1$ is that the last page requested in $\bH$ must be from level $\ell$; in particular, page $r_1$ (in level $\ell_1 > \ell$) is requested before some page in $L_\ell$. Let $\E$ denote this event; we bound
\begin{equation}\label{eq:cond-sketch}
    \Pr[X_\ell = 1 \mid Z_\ell = 1] \le \Pr[\E \mid Z_\ell = 1] = \sum_{t\ge 1} \Pr[\E \mid Z_\ell = 1, Y_\ell = t]\cdot \Pr[Y_\ell = t].
\end{equation}
Consider any fixed number of requests, say $t$. The conditional probability of any page in the head $\bH$ increases when we decrease the number of requests for pages in the tail $\bT$ among these $t$ requests. Therefore, 
\[
    \Pr[\E \mid Z_\ell = 1, Y_\ell = t] 
    \le \Pr[\E \mid Z_\ell = 0, Y_\ell = t]
    = \frac{1/2^{\ell_1}}{1-\tailp}\cdot t
    \le \frac{t}{2^{\ell_1-1}}, \text{ since } \tailp < \nf 12.
\]
Here, we used that given $Z_\ell = 0$, the probability that any request is for page $r_1$ is $\frac{1/2^{\ell_1}}{1-\tailp}$. The claim can now be established by putting together the bounds we have obtained in \eqref{eq:h1-main}.
\end{proof}

\subsection{Splitting the Head}

We have established bounds on the contribution of individual levels to $\sum_i p_i q_i$. We are now ready to put these together to bound the total contribution of the head $\bH$. 
To do this, we subdivide the head $\bH = [k]$ into three intervals $\bH_1 := [1, r_1]$, $\bH_2 = [r_1 + 1, r_2]$, and $\bH_3 := [r_2 + 1, k]$ as follows: we define $r_1$ to be the largest index $i$ such that $p_i \geq 4 e^2 \hsub \cdot \tailp$ (this is the same definition of $r_1$ that we used to define the level $\ell_1$), and $r_2$ to be the largest index $i$ such that $p_i \geq \tailp / 2^{\hsub}$. We will separately bound the contribution of each of these intervals to $\sum_{i=1}^k p_iq_i$.

\begin{restatable}[Bounding $\bH_1$]{lemma}{boundhone}
	\label{lem:h1}
	The contribution of $\bH_1$ is bounded as
	$\sum_{i \in \bH_1} p_i q_i \le 88 \hsub^2 \cdot \tailp$.   
\end{restatable}

\begin{proof}[Proof Sketch]
    This follows by using \Cref{cl:single-level} for a dominant level, if it exists, and level $\ell_1$, the last level in $\bH_1$, and using \Cref{cl:single-level-h1} for the remaining levels.
\end{proof}

\begin{lemma}[Bounding $\bH_2$]
	\label{lem:h2}
	The contribution of $\bH_2$ is bounded as
        $\sum_{i \in \bH_2} p_i q_i \le 12 \hsub^2 \cdot \tailp$.
\end{lemma}

\begin{proof} 
    By construction, the number of levels in $\bH_2$ is at most $\log (4e^2\hsub\cdot 2^{\hsub}) \le 6\hsub$.
    The claim now follows from \Cref{cl:single-level}.
\end{proof}

\begin{restatable}[Bounding $\bH_3$]{lemma}{boundhthree}
	\label{lem:h3}
	The contribution of $\bH_3$ is bounded as
	$\sum_{i \in \bH_3} p_i q_i \le 4\tailp$.
\end{restatable}

\begin{proof}[Proof Sketch]
    To bound the contribution of $\bH_3$, we note that $q_i \le 1$ 
    for all pages $i$ and $\sum_{i\in L_\ell} p_i$ for any level $\ell$
    is at most $n_\ell \cdot \frac{1}{2^\ell} \le 2^{\hsub}\cdot \frac{1}{2^\ell}$
    by \Cref{lem:dominate}. The lemma now follows by adding over all levels in $\bH_3$
    and using $p_i \le \tailp / 2^{\hsub}$ for any page $i\in \bH_3$.
\end{proof}

We conclude with the proof of the main theorem.

\begin{proof}[Proof of \Cref{thm:lru}]
	By adding up the bounds from \Cref{obs:tail} and \Cref{lem:h2,lem:h3,lem:h1}, we get 
     \[
	\sum_{i\in [n]} p_i q_i \le (100 \hsub^2 + 5) \cdot \tailp = O(\hsub^2 \tailp).
	\]
	By \Cref{lem:opt_lb}, this implies that \lru is $O(\hsub^3)$-competitive for stochastic caching.
\end{proof}

\section{The LFU Algorithm}\label{sec:lfu}In this section, we analyze the Least Frequently Used (LFU) algorithm for stochastic caching.  Recall that in LFU, we maintain a count for every page of the number of times that page was requested, and we keep in cache the $k$ pages with the highest count. When a page outside of cache is requested, we evict the lowest count page in order to serve the request, then revert back to the $k$ pages with highest count if necessary. Clearly, in the limit, the empirical count distribution will converge to the true distribution, and \lfu will eventually behave identically to \topk. However, the ``burn-in'' period until convergence could be arbitrarily long as a function of $n$; e.g., if $p_k - p_{k+1} = \epsilon$, then to even distinguish these two pages with probability $\delta$ requires $\Omega(\log(1 / \delta)/\epsilon^2)$ samples (e.g.,  \cite{canonne2022short}). 

Despite this, we show that \lfu achieves an additive loss of $O(\hsub \cdot k)$ in addition to the $O(\tailp)$ cost paid by \topk. Since $k$ is incurred as an initialization cost for any algorithm, we show:

\begin{theorem}
The \lfu algorithm has a competitive ratio $O(\hsub)$ for the stochastic caching problem.
\end{theorem}

The general proof idea is to show that every page in $\bH$ (i.e., top $k$) might be requested while it is outside of cache many times, but that after $O(\hsub)$ such cache misses, its count will forever be higher (from that point on, with good probability) than the count of any page with less than half its probability. If this is the case, we can charge any cache miss due to a request to $\bT$ to $T \cdot \tailp$. Each of the $k$ pages in $\bH$ will only miss $O(\hsub)$ times before the only pages displacing it from cache are ones with very similar probability, and hence distinguishing between these does not matter; the total contribution from these pages in $\bH$ is $O(k \hsub + T \cdot \tailp)$.  Below, we formalize this argument.
\begin{proof}
Define a parameter $\gamma = 16 \log (\min(k, n_{\max})) = O(\hsub)$.  For every $i\in \bH$ and $j \in [n]$ such that $p_j \leq p_i / 2$, define $\E_{ij}$ to be the bad event that if the first time $t$ that the empirical count of $i$ is at least $\gamma$, there is at least one time $t' \geq t$ such that the empirical count of $j$ exceeds the count of $i$. We show that $\E_{ij}$ is extremely unlikely for well separated $p_i$ and $p_j$.

\begin{restatable}[]{claim}{rwsingle}
\label{cl:rw-single}
    For all $i \in \bH$, $j \in [n]$ such that $p_j \leq p_i / 2^h$ for some $h \geq 1$, we have $\prob{\E_{ij}} \leq 2^{-h\gamma / 8}$.    
\end{restatable}

The proof uses standard facts and we leave it to \Cref{sec:lfu-proofs}. The last claim only gets easier if instead of considering a single page $j$ with $p_j \leq p_i/2^h$, we consider a group of pages whose cumulative probability is at most $p_i/2^h$, and ask what are the odds that any of these pages pass $i$ in count. (Since, even as we grouped these pages, their total count would still not exceed that of $i$ with the same probability as before.) 

\begin{corollary}
\label{cl:rw-set}
    For all $i \in \bH$ and any subset $S \subseteq [n]$ such that $\sum_{j \in S} p_j \leq p_i / 2^{h}$ for some $h \geq 1$, we have $\prob{\bigcup_{j \in S} \E_{ij}} \leq 2^{-h\gamma / 8}$.
\end{corollary}

For every $i\in \bH$, let $R_i$ be the number of requests to page $i$ until the count of $i$ never goes below that of any page $j$ with $p_j \leq p_i / 2$. Using \Cref{cl:rw-single} and \Cref{cl:rw-set}, we can show that the expected number of requests until we achieve this count separation can be upper bounded.

\begin{restatable}{claim}{rigamma}
    \label{cl:Ri_gamma}
    For every $i\in \bH \setminus D$ for some set $D$ such that $\sum_{i \in D} p_i = 2 \tailp$, we have $\expect{R_i} \leq  2\gamma$.
\end{restatable}

Note that if we had set $\gamma = \Theta(\log n)$, we could guarantee by a union bound over all pairs of pages $i$ and $j$ that every page $i \in \bH$ sees at most $\gamma$ requests before no page $j$ with $p_j \leq p_i/2$ ever has count higher than $i$ (with high probability). However, we can get away with smaller $\gamma$ using a more careful analysis, which we again defer to \Cref{sec:lfu-proofs}.  We now show how to conclude the proof.

Let $S_t$ be the subset of pages $i$ in $\bH$ that are both outside of cache at time $t$ and that have been requested at least $R_i$ times by time $t$. In other words, these are the pages $i\in \bH$ outside of cache at time $t$ such that no page $j$ with $p_j \leq p_i/2$ will have a count exceeding that of page $i$ after time $t$. Define an arbitrary perfect matching $M$ between the pages of $S_t$ and a subset of the pages from $\bT$ in cache at time $t$. If page $i \in S_t$ is matched to page $M(i)$, then $p_{M(i)} \geq p_i/2$, since the count of $M(i)$ must be higher than that of $i$, and $i \in S_t$.

Let $C^t$ be the $k$ pages with highest counts at time $t$. Then we can write the total cost of the \lfu algorithm as
\begin{align}
    \sum_t \sum_{i =1}^n p_i \prob{i \not \in C^t} &\leq \sum_t \sum_{i \not \in S_t} p_i \prob{i \not \in C^t} + \sum_t \sum_{i \in S_t} p_i \prob{i \not \in C^t} + \sum_t \tailp \nonumber \\
    &\leq  2 T \cdot \tailp + \sum_{i \in \bH} \expect{R_i} +  \sum_t \sum_{i \in S_t} p_i \prob{i \not \in C^t} + T \cdot \tailp \label{line:use_Ri} \\
    &\leq 2\gamma \cdot k + 2 \cdot \sum_t \sum_{i \in S_t} p_{M(i)} + 3T \cdot \tailp \label{line:matchtotail}\\
    &\leq 2\gamma \cdot k + 5T \cdot \tailp. \label{line:totalmatching}
\end{align}
Step \eqref{line:use_Ri} follows because the expected number of misses for each page $i \in \bH$ until it belongs to $S_t$ is $\expect{R_i}$, step \eqref{line:matchtotail} follows by \Cref{cl:Ri_gamma} and since $p_i \leq 2 p_{M(i)}$, and step \eqref{line:totalmatching} follows since the matched pages in cache form a subset of $\bT$, and hence their total probability is less than $\tailp$.
\end{proof}

\section{Closing Remarks}

In this work, we study the stochastic caching problem as a way to move beyond worst-case competitive ratio analysis for online algorithms.  We give new evidence for the effectiveness of the popular \lru and \lfu caching policies, by showing that for many distributions, they incur approximation factors significantly better than classical bounds predict. En route, we identify subset entropy as an important parameter of the distribution governing the achievable competitive ratio of stochastic caching; this information-theoretic quantity might be of independent interest. 

There are several avenues for future work. 
The first obvious direction is to either improve the competitive ratio of \lru from $O(\hsub^3)$ to $O(\hsub)$, or show that this is not possible. In \Cref{sec:lru_lb}, we show that there are distributions for which the cost of LRU is $\Omega(\hsub \cdot p(\bT))$, meaning that our techniques can at best prove an $O(\hsub^2)$ bound.

Another direction is to generalize the result to more complex stochastic models, e.g., the Markov paging model \cite{KarlinPR00} (this would entail developing an analogous quantity to subset entropy). 
Going beyond caching, we often prove lower bounds for randomized online algorithms by exhibiting pathological input distributions. Parameterizing what makes these distributions difficult and, more generally, understanding and quantifying the performance of an online algorithm through a fine-grained measure such as subset entropy are promising directions of future research.

\section*{Appendix}

\appendix

\section{Missing Proofs from \Cref{sec:prelim}}\label{sec:prelim_proofs}
\begin{proof}[Proof of \Cref{lem:dominate}]
    Note that for every $\ell$, the distribution $\cD_{S_{\ell}}$ is uniform up to a factor of 2 (i.e., the maximum probability is at most twice the minimum probability). As noted, the Shannon entropy of such a distribution is $\Theta(\log n_{\ell})$, where $n_{\ell}$ is the size of its support.  We choose $S$ as follows: if $n_{\max} \le k$, then $S = \arg\max_\ell |S_\ell|$ and hence $|S| = n_{\max}$; else, $S$ is any $k$-sized subset of $\arg\max_\ell |S_\ell|$ and hence $|S| = k$. Then, we get
    \[
        \hsub(\cD, k) \ge \ent(\cD_S) = \Omega(\log |S|) = \Omega(\log (\min \{k, n_{\max} \})).
    \]

    We now show the other direction, by proving that for any distribution $\cD'$ with $\max_\ell |\{u\in U: 1/2^{\ell+1} < \cD'(u) \le 1/2^{\ell})\}| \leq N$, we have that $H(\cD') \leq O(\log N))$. Note that $N \le \min\{2n_{\max}, k\}$ for $\cD_S$, where $S \subseteq [n]$ is any subset of cardinality at most $k$, since pages in $S$ whose probabilities in $\cD_S$ differ by at most a factor of two can come from at most two adjacent levels of the original distribution $\cD$.

    Letting $S'_\ell := \{u: 1/2^{\ell+1} < \cD'(u) \le 1/2^\ell\}$, we have,
    \[
        \ent(\cD') 
        = -\sum_{u} \cD'(u) \cdot \log \cD'(u)
        = -\sum_\ell \sum_{u\in S'_\ell} \cD'(u) \cdot \log \cD'(u).
    \]
    Now, for any $u\in S'_\ell$, we have $\cD'(u) \le \nf{1}{2^\ell}$ and $-\log \cD'(u) = \log \nf{1}{\cD'(u)} < \log 2^{\ell+1} = \ell+1$. 

    Therefore,
    $\ent(\cD') 
        \le \sum_\ell n_\ell\cdot \frac{\ell+1}{2^\ell}$.
    where $n_\ell := |S'_{\ell}|$.
    
    To bound the RHS of this inequality, we use a factor-revealing LP.     
    We note that $\ent(\cD')$ is at most the optimal primal objective in the following LP:
    \[
        \text{maximize } \sum_{\ell \ge 1} n_\ell\cdot \frac{\ell+1}{2^\ell} 
        \mbox{ subject to } 
                 \sum_\ell \frac{n_\ell}{2^\ell} \le 1
                \mbox{ and } 
            \{ n_\ell \le N\}_{\forall \ell}.
\]
The dual LP is
    \[
        \text{minimize } \alpha + N \sum_\ell \beta_\ell, \text{ subject to }
            \left\{ \frac{\alpha}{2^\ell} + \beta_\ell \ge \frac{\ell+1}{2^\ell} \right\}_{\forall  \ell}.
        \]
    Now, a feasible dual solution is given by:
    \[
        \alpha = \log N \quad\quad
        \beta_\ell = \begin{cases}
                    \frac{\ell+1}{2^\ell} & \text{if } \ell > \log N -1\\
                    0 & \text{if } \ell \le \log N-1.
                    \end{cases}
    \]
    The dual objective for this solution is $O(\log N) = O(\log (\min\{k, n_{\max}\}))$. Therefore, $\hsub(\cD, k) = \ent(\cD_S) \le O(\log (\min \{k, n_{\max} \})) = O(\log (\min \{k, \max_\ell |S_\ell| \}))$.
\end{proof}

\section{Missing Proofs from \Cref{known distribution}}\label{sec:topk_proofs}
We first define our true dual vector $y$. For any time $t$, let $G(t) \subseteq[n]$ be the set of pages $i \not \in [k]$, for which the interval $I(i,\rho(i,t))$ has $|I(i,\rho(i,t))| \leq 3 / p_i$. Let $\sigma^t_1, \ldots, \sigma^t_{m^t}$ be the pages of $G(t)$ in decreasing order of $p_i$. Then define $S^t_j = [k] \cup \{\sigma^t_1, \ldots, \sigma^t_j\}$, and 
\[y_{S}^t = \begin{cases}
    p_{\sigma^t_j} - p_{\sigma^t_{j+1}} & \text{if } S = S^t_j \text{ and } j < m^t. \\
    p_{\sigma^t_{m^t}} & \text{if } S = S^t_j \text{ and } j = m^t. \\
    0 & \text{otherwise}.
\end{cases}\]
In words, we define dual variables only for prefixes of the pages in decreasing order of probability (after removing pages whose intervals exceed their expected length by a constant), and include the $k$ pages in $\bH$ in all these prefixes.

We argue that this dual is approximately feasible with constant probability, and that its objective value is sufficiently large. We start with the first claim. Let $n_{\max} := \max_\ell |S_\ell|$ in \Cref{lem:dominate} and let $\eta := \min \{n_{\max}, k\}$. Recall that by \Cref{lem:dominate}, we have that $\hsub = \Theta(\log \eta)$.

\begin{lemma}
    \label{lem:dual_feas}
    Fix a time step $t$. With probability at least $3/4$, for all $i \in [n]$, the solution $y$ satisfies the dual constraints associated with $(i,\rho(i,t))$ up to a factor of $3 \ln (2\eta)$.
\end{lemma}

\begin{proof}
    Since $t$ is fixed, for convenience we define $I(i) := I(i,\rho(i,t))$.

    For pages $i \in [n] \setminus ([k] \cup G(t))$, for all $t' \in I(i)$ and $S$ such that $i \in S$, note that $y_S^{t'} = 0$. So, the corresponding dual constraint is satisfied vacuously.

    Next, consider constraints associated with pages $i \in G(t)$. For all $t' \in I(i)$, if $j^*(t')$ is the index of $i$ in $\sigma^{t'}_1, \ldots, \sigma^{t'}_{m^{t'}}$, then by design
    \begin{align*}\sum_{S: \, i \in S} y_S^{t'} &= \sum_{j=j^*(t)}^{m^{t'}-1} \left(p_{\sigma^{t'}_j} - p_{\sigma^{t'}_{j+1}}\right) + p_{\sigma^{t'}_{m^{t'}}} = p_{\sigma^{t'}_{j^*(t')}} = p_i,
    \intertext{which means}
        \sum_{t' \in I(i)} \sum_{S :\, i \in S} y_S^{t'} &= p_i \cdot |I(i)| \leq p_i \cdot \frac{3}{p_i} = 3.
    \end{align*}
    It remains to handle constraints associated with $i \in [k]$. This time, for all $t' \in I(i)$, we have
    \begin{align*}\sum_{S: \, i \in S} y_S^{t'} &= \sum_{j=1}^{m^{t'}-1} \left(p_{\sigma^{t'}_j} - p_{\sigma^{t'}_{j+1}}\right) +  p_{\sigma^{t'}_{m^{t'}}} = p_{\sigma^{t'}_{1}} \leq p_k,
    \end{align*}
    and hence
    \[
    \prob{\sum_{t' \in I(i)} \sum_{S: \, i \in S} y_S^{t'} 
 \geq 3 \ln (2\eta)} \leq \prob{|I(i)| \cdot p_k \geq 3 \ln (2\eta)} 
 \leq \prob{|I(i)| 
 \geq E[|I(i)|] \cdot 3 \ln (2\eta) \cdot \frac{p_{i}}{p_k}}.
\]
 Since $|I(i,j)|$ is a geometric random variable with parameter $p_i$, we get that
 \[
 \prob{\sum_{t' \in I(i)} \sum_{S: \, i \in S} y_S^{t'} \geq 3 \ln(2\eta)}
 \leq (1-p_i)^{E[|I(i)|] \cdot 3 \ln (2\eta) \cdot \frac{p_{i}}{p_k}} 
 \leq e^{-3 \ln (2\eta) \cdot \frac{p_{i}}{p_k}}
 \leq (2\eta)^{-3\cdot \frac{p_{i}}{p_k}}.
\]

Let $L_{\ell}$ be the pages $i$ in $[k]$ such that $p_k \cdot 2^{\ell+1} > p_i \ge p_k \cdot 2^\ell$. Note that $|L_\ell| \le 2 n_{\max}$ since it contains pages from at most two sets $S_{\ell'}, S_{\ell'+1}$ for some $\ell'$. Moreover, $|L_\ell| \le k$. Hence, $|L_\ell| \le 2\eta$. 

Now, let $\E$ be the event that there exists an $i \in [k]$ for which $\sum_{t' \in I(i)} \sum_{S: \, i \in S} y_S^{t'} 
 \geq 3 \ln (2\eta)$, i.e., the dual constraint for page $i$ is violated beyond a factor of $3\ln (2\eta)$. By a union bound, 
\begin{align*}
    &\prob{\E}
    \leq \sum_{\ell \ge 0} \sum_{\substack{i \in [k]: \\ i \in L_\ell}} 
    (2\eta)^{-3 \cdot \frac{p_i}{p_k}} 
    \leq \sum_{\ell \ge 0} \sum_{\substack{i \in [k]: \\ i \in L_\ell}} 
    (2\eta)^{-3 \cdot 2^\ell} 
    \leq \sum_{\ell \ge 0} 2\eta\cdot (2\eta)^{-3 \cdot 2^\ell} 
    \leq (2\eta)^{-2}
    \leq \frac{1}{4}.\qedhere
    \end{align*}
\end{proof}

\begin{lemma}
    \label{lem:dual_profit}
    The expected objective of the dual $y$ at time $t$ is at least $2\tailp / 3$.
\end{lemma}

\begin{proof}
    By construction, the dual objective at time $t$ is 
    \begin{align*}\sum_{S} (|S| - k) \cdot y_S^t &= m^t \cdot p_{\sigma_{m^t}^t} + \sum_{j=1}^{m^t-1} j \cdot (p_{\sigma_{j}^t}  - p_{\sigma_{j+1}^t}) = \sum_{i \in G(t)} p_{i}.\end{align*}

    By a Markov bound, every page $i \in [n] \setminus [k]$ appears in $G(t)$ with probability at least $2/3$, hence the claim.
\end{proof}

Using these, we can conclude the proof of \Cref{lem:opt_lb}, which in turn implies \Cref{thm:ratio} by weak duality.

\begin{proof}[Proof of \Cref{lem:opt_lb}]
    From \Cref{lem:dual_feas} and \Cref{lem:dual_profit}, there exists a dual solution $y$ such that:
    \begin{enumerate}[(a)]
        \item At every time $t$, with probability $\geq 3/4$, solution $y$ simultaneously satisfies all constraints $(i,j)$ such that $t \in I(i,j)$ up to a multiplicative factor of $3 \ln(2 \eta)$.
        \item The expected dual objective value of this solution $y$ is $2\tailp/3$.
    \end{enumerate}

    Define a modified dual solution $\widehat y$ as follows. For every time $t$, if there exists a dual constraint $(i,j)$ such that $t \in I(i,j)$ that is violated by a multiplicative factor greater than $3 \ln(2\eta)$ (call this event $\E(t)$), then set $\widehat y_S^t = 0$ for all $S \subseteq [n]$. Otherwise, set $\widehat y_S^t = y_S^t / (3\ln(2\eta))$.

    The new dual is always feasible. Furthermore, when $y^t_S > 0$, we have $\widehat y^t_S = y^t_S / (3\ln(2\eta))$ with probability at least $3/4$. Hence, by the linearity of expectation, the dual objective of $\widehat y$ is
      \begin{align}
        \sum_{S} (|S| - k) \cdot \expect*{\widehat y^t_S} &\geq \frac{1}{3 \ln(2\eta)} \sum_S (|S| - k) \cdot \left(\expect*{y_S^t} - \expect*{y_S^t |  \E(t)} \cdot \prob{\E(t)}\right) \nonumber\\
        &\geq
        \frac{2\tailp/3 - \tailp/4}{3 \ln (2\eta)} \label{eq:ub_on_yts}
         \\
         &\geq \frac{5\tailp / 12}{3 \ln (2\cdot \min \{k, n_{\max} \})} \nonumber,
    \end{align}
    where in \eqref{eq:ub_on_yts} we used $\sum_{S} (|S| - k) \cdot y_S^t \leq \tailp$ with probability $1$, together with \Cref{lem:dual_feas,lem:dual_profit}.
    
    The characterization of $\hsub = \Theta(\log \min \{k, n_{\max}\})$ from \Cref{lem:dominate} concludes the proof.
\end{proof}

\section{Missing Proofs from \Cref{sec:unknown2}}\label{sec:lru-proofs} 
\subsection{Proof of \Cref{lem:prop}} \label{Lemma 4.4}

\begin{proof}[Proof of \Cref{lem:prop}]
We prove the three parts separately.

{\em (i)} It follows from the fact that $0\le X_\ell \le n_\ell$, and the fact that $\EE[X_\ell] = \sum_{i\in L_\ell} q_i$. 

    {\em (ii)}
    This fact follows from standard coupon collector bounds, since each request is for a page in $L_\ell$ with probability $\ge \nf{1}{2^{\ell+1}}$ and there are $n_\ell$ distinct pages in $L_\ell$. 

    {\em (iii)} 
    Since $\ell$ is a non-dominant level, even after conditioning on the page requests being from level $\ell$, the probability that any request is for a page in $\TT$ is at most $2\tailp$. We use this fact in the rest of this proof.
    
    For $t = 1$, we have
\begin{align*}
    \Pr[Z_\ell \ge 1] &\le \EE[Z_\ell] = \sum_{L \ge 1} \EE[Z_\ell \mid Y_\ell = L]\cdot \Pr[Y_\ell = L] \le  2\tailp \sum_{L \ge 1} L \cdot  \Pr[Y_\ell = L] \\
    &= 2 \tailp \cdot \EE[L] \le \tailp \cdot 2^{\ell+1}\cdot \hsub.
\end{align*}

For $t\ge 2$, let $\mu := \EE[Y_\ell] \le 2^{\ell+1}\cdot \hsub$. Now, the value of $Y_\ell$ can be partitioned into intervals $(e(\gamma-1)\mu, e\gamma\mu]$ for integers $\gamma \ge 1$. Using this partitioning, we get 

\begin{align*}
    \Pr[Z_\ell \ge t] 
    &= \sum_{\gamma \ge 1} \Pr[Z_\ell \ge t \mid e(\gamma -1)\mu < Y_\ell \le e \gamma \mu]\cdot \Pr[e(\gamma -1)\mu < Y_\ell \le e \gamma \mu]  \\
    &= \sum_{\gamma \ge 1} \Pr[Z_\ell \ge t \mid Y_\ell \le e \gamma \mu]\cdot \Pr[Y_\ell > e(\gamma -1)\mu].
\end{align*}
Using the Binomial distribution, we get $\Pr[Z_\ell \ge t \mid Y_\ell \le e \gamma \mu]
    \le {e\gamma\mu \choose t} (2\tailp)^t \le \left(\frac{e^2\gamma\mu}{t}\right)^t (2\tailp)^t$,
where the last step used Stirling's approximation.

We now bound the second term, $\Pr[Y_\ell > e(\gamma -1)\mu]$. Observe that for $Y_\ell > e(\gamma-1)\mu$ to hold, each of the first $\gamma-1$ blocks of $e\mu$ requests must fail to request all the pages from level $\ell$. By Markov's inequality, since $\mu = \EE[Y_\ell]$, each block of $e\mu$ requests fails to request all pages from level $\ell$ with probability at most $1/e$, and since these failure events are independent, $\prob{Y_\ell > e(\gamma-1)\mu} \leq e^{-(\gamma-1)}$. Therefore, for $\gamma \ge 1$, we have $\Pr[Y_\ell > e(\gamma -1)\mu] \le e^{-(\gamma-1)}$.
Putting these together, we get
\begin{align}
    \Pr[Z_\ell \ge t]
    &\le \sum_{\gamma \ge 1} \left(\frac{e^2\gamma\mu}{t}\right)^t (2\tailp)^t \cdot e^{-(\gamma -1)}\nonumber\\
    &\le \left(\frac{e^2\mu}{t}\right)^t \cdot (2\tailp)^t \cdot \sum_{\gamma = 1}^{\infty} \frac{\gamma^t}{e^{\gamma-1}} 
    \leq \enspace \left(\frac{e^2\mu}{t}\right)^t \cdot (2\tailp)^t \cdot \left(\sum_{\gamma = 1}^{2t \ln(t) -1} \frac{\gamma^t}{e^{\gamma - 1}} + \sum_{\gamma = 2t \ln t}^{\infty} \frac{\gamma^t}{e^{\gamma-1}}\right) \nonumber\\
    &\le \left(\frac{e^2\mu}{t}\right)^t \cdot (2\tailp)^t \cdot \left((2t \ln (t) -1 ) \cdot \frac{t^t}{e^{t-1}} + 1\right) \label{eq:gamma_bound} \\
    &\le 
    \left(2e^2\mu\right)^t \cdot \tailp^t.  \tag{since $t\ge 2$}
\end{align}
Step \eqref{eq:gamma_bound} follows since the expression $\gamma^t / e^{\gamma-1}$ is maximized at $\gamma = t$, so the first partial sum can be bounded by its length times the maximum $t^t/e^{t-1}$. The suffix sum converges to less than one.
\end{proof}

\subsection{Proof of \Cref{cl:single-level-h1}} \label{Claim 4.6}

\begin{proof}[Proof of \Cref{cl:single-level-h1}]
    
    Recall that by \eqref{eq:indicator}, we have that for any level $\ell$,
$\sum_{i\in L_\ell} p_i q_i \le \frac{1}{2^\ell} \cdot \EE[X_\ell]$,
where $X_\ell$ denotes the number of pages in $L_\ell$ that do not appear among the first $k$ distinct pages. We proceed to bound $\EE[X_\ell]$. By definition, $X_\ell \le Z_\ell$. Therefore, we can write 
\begin{align}
    \EE[X_\ell] 
    &= \sum_{t\ge 1} t\cdot \Pr[X_\ell = t]\nonumber = \Pr[X_\ell = 1] + \sum_{t\ge 2} t\cdot \Pr[X_\ell = t] = \Pr[X_\ell = 1] + \sum_{t\ge 2} \Pr[X_\ell \ge t]
    \nonumber\\
    &\le \Pr[X_\ell = 1 \mid Z_\ell = 1]\cdot \Pr[Z_\ell = 1] + \sum_{t\ge 2} \Pr[Z_\ell \ge t]\nonumber\\
    &\le \Pr[X_\ell = 1 \mid Z_\ell = 1]\cdot \Pr[Z_\ell \ge 1] + \sum_{t\ge 2} \Pr[Z_\ell \ge t].\label{eq:xell}
\end{align}

We can bound factors of \eqref{eq:xell} of the form $\Pr[Z_\ell \ge t]$ using \eqref{eq:tail-alt}. We get
\begin{align*}
    \sum_{t\ge 2} \Pr[Z_\ell \ge t]
    &\le \sum_{t\ge 2} \left(2e^2 \cdot \tailp \cdot 2^{\ell+1}\cdot \hsub\right)^t \le 2 \left(2e^2 \cdot \tailp \cdot 2^{\ell+1}\cdot \hsub\right)^2,
\end{align*}
where in the last step, we used 
\begin{align*}
    2e^2 \cdot \tailp \cdot 2^{\ell+1}\cdot \hsub \le 2e^2 \cdot \tailp \cdot 2^{\ell_1}\cdot \hsub \le \frac{2e^2 \cdot \tailp \cdot \hsub}{4e^2\hsub\cdot \tailp} = \nf 12.
\end{align*}    
Thus, $\sum_{t\ge 2} \Pr[Z_\ell \ge t]
    \le 2 \left(2e^2 \cdot \tailp \cdot 2^{\ell+1}\cdot \hsub\right)^2
    \le \frac{\hsub \cdot \tailp}{2^{\ell_1-2\ell-6}}$, 
since $\tailp \leq 1/(4e^2\hsub\cdot 2^{\ell_1})$.

It remains to bound the first part of \eqref{eq:xell}. From \eqref{eq:tail-alt-1}, we have that $\Pr[Z_\ell \ge 1] \le 2\tailp\cdot 2^{\ell+1}\cdot \hsub$. We are left to bound $\Pr[X_\ell = 1 \mid Z_\ell = 1]$. Since there is only one request from $\bT$ before all pages in level $\ell$ are requested, a necessary event for $X_\ell$ to be $1$ is that the last page requested in $\bH$ must be from level $\ell$. In particular, consider page $r_1$ which is in level $\ell_1$ and therefore not in level $\ell < \ell_1$. Let $\E$ denote the event that page $r_1$ is requested before all pages in level $\ell$ are requested. Then,
\begin{align}
    \Pr[X_\ell = 1 \mid Z_\ell = 1] 
    &\le \Pr[\E \mid Z_\ell = 1] = \sum_{t\ge 1} \Pr[\E \mid Z_\ell = 1, Y_\ell = t]\cdot \Pr[Y_\ell = t].\label{eq:cond}
\end{align}
Consider any fixed number of requests, say $t$. The conditional probability of any page in the head $\bH$ increases when we decrease the number of requests for pages in the tail $\bT$ among these $t$ requests. Therefore, $\Pr[\E \mid Z_\ell = 1, Y_\ell = t] 
    \le \Pr[\E \mid Z_\ell = 0, Y_\ell = t]$.
    
Now, let $I_t(\E)$ denote the number of requests for page $r_1$ among the first $t$ requests. Given that $Z_\ell = 0$, the probability that any request is for page $r_1$ is $2^{-\ell_1}/(1-\tailp)$. Therefore, 
\begin{align*}
    \Pr[\E \mid Z_\ell = 1, Y_\ell = t]
    &\le \EE[I_t(\E) \mid Z_\ell = 0] = \frac{1/2^{\ell_1}}{1-\tailp}\cdot t \le \frac{t}{2^{\ell_1-1}}, 
\end{align*}
since $\tailp < \nf 12$.
Substituting in \eqref{eq:cond}, we get
\begin{align*}
    \Pr[X_\ell = 1 \mid Z_\ell = 1] &= \sum_{t\ge 1} \Pr[\E \mid Z_\ell = 1, Y_\ell = t]\cdot \Pr[Y_\ell = t] \le \sum_{t\ge 1} \frac{t}{2^{\ell_1-1}} \cdot \Pr[Y_\ell = t] \\
    &\le \frac{\EE[Y_\ell]}{2^{\ell_1-1}} \le \frac{\hsub\cdot 2^{\ell+1}}{2^{\ell_1-1}} = \frac{4 \hsub}{2^{\ell_1-\ell}},
\end{align*}
where we again used \eqref{eq:Yell_bound} in the last inequality. Therefore,
\begin{align*}
    \Pr[X_\ell = 1 \mid Z_\ell = 1]\cdot \Pr[Z_\ell \ge 1]
    \le \frac{4 \hsub}{2^{\ell_1-\ell}} \cdot 2\tailp\cdot 2^{\ell+1}\cdot \hsub = \frac{\hsub^2\cdot \tailp}{2^{\ell_1-2\ell-4}}.
\end{align*}
This bounds the first expression in \eqref{eq:xell}.

Substituting our bounds for both expressions in \eqref{eq:xell}, we get
 \begin{align*}
    \EE[X_\ell] \le \frac{\hsub^2\cdot \tailp}{2^{\ell_1-2\ell-4}} + \frac{\hsub \cdot \tailp}{2^{\ell_1-2\ell-6}} \le \frac{80 \hsub^2\cdot \tailp}{2^{\ell_1-2\ell}}.
\end{align*}
Therefore, 
\begin{align*}
    \sum_{i\in L_\ell} p_i q_i \le \frac{1}{2^\ell} \cdot \EE[X_\ell] \le \frac{1}{2^\ell}  \cdot \frac{80 \hsub^2\cdot \tailp}{2^{\ell_1-2\ell}} = \frac{80 \hsub^2\cdot \tailp}{2^{\ell_1-\ell}}.
& \qedhere
\end{align*}
\end{proof}

\subsection{Proof of \Cref{lem:h1}} \label{Lemma 4.7}

\begin{proof}[Proof of \Cref{lem:h1}]
    We use \Cref{cl:single-level} to bound both a dominant level, if it exists, and the last level in $\bH_1$, i.e., $\ell_1$. For both these levels $\ell = \ell^*$ and $\ell = \ell_1$, we get $
        \sum_{i\in L_\ell} p_i q_i \le 4\hsub\cdot \tailp$,
    by \Cref{cl:single-level}.

    For non-dominant levels $\ell < \ell_1$, we use \Cref{cl:single-level-h1}. Thus,
    $\sum_{\ell < \ell_1} \sum_{i\in L_\ell} p_i q_i
        \le \sum_{\ell_1-\ell \ge 1} \frac{80 \hsub^2\cdot \tailp}{2^{\ell_1-\ell}} 
        \le 80 \hsub^2\cdot \tailp$.
    Summing the bounds, we get $
        \sum_{i\in \bH_1} p_i q_i
        \le 4\hsub\cdot \tailp + 80 \hsub^2\cdot \tailp
        \le 88 \hsub^2 \cdot\tailp$.
\end{proof}

\subsection{Proof of \Cref{lem:h3}} \label{Lemma 4.9}

\begin{proof}[Proof of \Cref{lem:h3}]
    Let $\ell_3$ be the first level $\ell$ that contains pages from $\bH_3$, i.e., all pages in $\bH_3$ are in levels $\ell_3$ and above. Then,
    $\sum_{i\in \bH_3} p_i q_i
        \le \sum_{\ell \ge \ell_3} \frac{1}{2^\ell} \sum_{i\in L_\ell} q_i
        \le \sum_{\ell \ge \ell_3} \frac{1}{2^\ell} \cdot n_\ell$,
    since $q_i \le 1$ by virtue of being a probability. 
    
    By \Cref{lem:dominate}, for any level $\ell$, the number pages $n_\ell$ is at most $2^{\hsub}$. 
    Therefore,
    \[
        \sum_{i\in \bH_3} p_i q_i 
        \le \sum_{\ell \ge \ell_3} \frac{1}{2^\ell} \cdot n_\ell
        \le 2^{\hsub} \cdot \sum_{\ell \ge \ell_3} \frac{1}{2^\ell}
        \le 2^{\hsub} \cdot \frac{1}{2^{\ell_3-1}}.
    \]

    Note that $p_i \le \tailp / 2^{\hsub}$ for any page $i\in \bH_3$. If page $i$ is in level $\ell$, then $p_i \ge \nf {1}{2^{\ell+1}}$. Combining the two, we get 
    $\frac{1}{2^{\ell+1}} \le p_i \le \frac{\tailp}{2^{\hsub}}$,  and, in particular, 
        $\frac{1}{2^{\ell_3-1}} \le \frac{4\tailp}{2^{\hsub}}.$ Putting these together, we get
    $\sum_{i\in \bH_3} p_i q_i
        \le 2^{\hsub} \cdot \frac{1}{2^{\ell_3-1}}
        \le 4\tailp$.
\end{proof}

\section{Missing Proofs from \Cref{sec:lfu}}\label{sec:lfu-proofs} 
\subsection{Proof of \Cref{cl:rw-single}} \label{Claim 5.1}

\begin{proof}[Proof of \Cref{cl:rw-single}]
    Fix $i \in \bH$ and $j$ such that $p_j \leq p_i / 2$. 
    Consider the following random walk on the line that starts at $0$, moves $+1$ with probability $p_i/(p_i + p_j) \geq 1-1/(2^h+1)$, and moves $-1$ with probability $p_j/(p_i + p_j) \leq 1/(2^h+1)$.  After $\gamma$ time steps, the position is at least $\gamma/2$ with probability at least $1 - 2^{-h \gamma/4}$.  Indeed, observe that the probability the walk is at position $<\gamma/2$ after $\gamma$ time steps is the probability that $\gamma$ i.i.d. Bernoullis $X_1, \ldots, X_\gamma$ with probability $q = 1/(2^h + 1)$ sum to $\gamma/4$. The sum $X = \sum_i X_i$ has mean $\mu = q\gamma$ and if we set $\delta = \gamma/(4\mu) -1$, by a Chernoff bound,  
    \begin{align*}
        \prob{\sum_i X_i \geq \frac{\gamma}{4}} = \prob{\sum_i X_i \geq (1+\delta)\mu} \leq \frac{e^{(1+\delta)\mu}}{(1+\delta)^{(1+\delta)\mu}} \leq e^{\gamma/4}(4q)^{-\gamma/4} \leq 2^{-h\gamma/4}.
    \end{align*}

    A random walk with bias $1 - 1/(2^h + 1)$ at position $z$ returns to $0$ with probability exactly $2^{-hz}$ (see, e.g.,~\cite[page 6]{steele2001stochastic}). 
    
    Now, it is easy to see that the event $\E_{ij}$ has the same probability as that of this random walk returning to $0$ after $\gamma$ time steps. Therefore, if $\E_{ij}$ holds either (1) the random walk was at position $\leq \gamma/4$ after $\gamma$ steps, or (2) starting at position $\gamma/4$, the walk returned to $0$ after a finite amount of time. By a union bound, $\prob{\E_{ij}} \leq 2^{-h\gamma / 8}$.
\end{proof}

\subsection{Proof of \Cref{cl:Ri_gamma}} \label{Claim 5.3}

\begin{proof}[Proof of \Cref{cl:Ri_gamma}]
    Recall that $\gamma = 16 \hsub = 16\log \min(k, n_{\max})$. We do a case analysis.

    \textbf{Case 1.} Suppose that $k \leq n_{\max}$, so $\gamma = 16 \log k$. Recalling that the pages of $\bT$ are ordered by decreasing probability,  define $G_\ell := \{i \in \bT \mid i \geq 2k+1, \ i = \ell \mod k\}$. 
    Define $q := 2\cdot \sum_{i \in G'_0} p_i$, and let $I = \{i \in \bH \mid p_i \geq q\}$.

    Note that the total probability mass of $\bH \setminus I$ is at most $k \cdot q = O(\tailp)$. Also, for all $i \in I$ and all $\ell$, we have $p_i \geq 2 \sum_{j \in {G_\ell}} p_j$. By \Cref{cl:rw-set}, after $\gamma$ requests to $i$, the probability any page in $G_\ell$ ever has count exceeding that of $i$ is $2^{-\gamma} = 1/\poly(k)$. By the same argument, the probability any page in $\{k+1, \ldots, 2k\}$ exceeds the count of $i$ after $\gamma$ requests to $i$ is also $1/\poly(k)$. By a union bound, no page in the tail exceeds the count of $i$ after $\gamma$ requests to $i$ with probability $1-\poly(k)$, and hence the expected number of requests to $i$ before this condition holds is at most $ 2 \gamma$.     
    
    \textbf{Case 2.} Suppose now that $n_{\max} \leq k$, so $\gamma = 16 \log n_{\max}$. Then let $L_\ell = \{j \in \bT \mid p_i \in [p_{k} 2^{\ell}, p_{k} 2^{\ell-1}) \}$. By definition $|L_\ell| \leq n_{\max}$. For every $i \in \bH$ and every $\ell \geq 1$, by \Cref{cl:rw-set} after $\gamma$ requests to $i$, a page $j$ in $L_\ell$ will have count exceeding $i$ with probability no more than $2^{-h\gamma / 8} = (n_{\max})^{-2 \ell}$. By a union bound, no page in $L_\ell$ has count exceeding $i$ with probability $(n_{\max}^{-\ell}) \leq 2^{-\ell - 1}$, and hence by another union bound, no page in $\bT$ with probability more than $p_i/2$ will have count exceeding $i$ with probability at least $1/2$. Once again, this means the expected number of requests to $i$ before this condition holds is $2 \gamma $.
\end{proof}

\section{Computing Subset Entropy}
\label{sec:subset-algo}

\newcommand{\Smax}{{S_{\max}}}
\newcommand{\pmax}{p_{\max}}
\newcommand{\pmin}{p_{\min}}

We recall the definition of $k$-subset entropy \eqref{eq:subsetH}. We show the following lemma, which immediately implies an efficient algorithm for computing $\hsub(\cD, k)$ (we call this the interval property of $\hsub(\cD, k)$):
\begin{lemma}[Interval Property of $\hsub(\cD, k)$]\label{lem:interval}
    Consider a discrete distribution $\cD: [n] \rightarrow [0, 1]$ such that $\cD(i) \ge \cD(j)$ for any $i \ge j$. Let $k \le n$ be any positive integer. Then, there is always a subset $S \subseteq [n], |S| \le k$ that maximizes $\ent(\cD_S)$ among subsets of $[n]$ of cardinality at most $k$ which is of the form $S = \{r, r+1, \ldots, t\}$.
\end{lemma}

Next, we prove this lemma. An important parameter in our proof will be the weighted geometric mean defined as $\tau = \prod_{i\in S} p_i^{q_i}$. 

\begin{claim}\label{cl:tau-opt}
    The weighted geometric mean $\tau = \prod_{i\in S} p_i^{q_i}$ satisfies the following condition:
    \[
        \frac{\partial \ent(S)}{\partial p_j} 
        \begin{cases}
            < 0 & \text{ if } p_j > \tau\\ 
            = 0 & \text{ if } p_j = \tau\\
            > 0 & \text{ if } p_j < \tau.
        \end{cases}
    \]   
\end{claim}
\begin{proof}
    Let us use the notation $p_i := \cD(i)$, i.e., $p_1 \ge p_2 \ge \dots \ge p_n$. 
    We have:
    \[
        \ent(S) 
        = - \frac{\sum_{i\in S} p_i \log \frac{p_i}{\sum_{i\in S} p_i}}{\sum_{i\in S} p_i}
        = - \frac{\sum_{i\in S} p_i \log p_i}{\sum_{i\in S} p_i} + \log \left(\sum_{i\in S} p_i\right).
    \]
        Note that the entropy of the conditional distribution on $S$ given by $\ent(\cD_S)$ is a function of the values $p_i, i\in S$. 
    We consider the partial derivative of this function with respect to any specific $p_j, j\in S$:
    \begin{align*}
        \frac{\partial \ent(S)}{\partial p_j}
        &= - \frac{\left(\sum_{i\in S} p_i\right) (1+\log p_j)  - \sum_{i\in S} p_i \log p_i}{\left(\sum_{i\in S} p_i\right)^2} + \frac{1}{\sum_{i\in S} p_i} \\
        & = - \frac{\log p_j}{\sum_{i\in S} p_i} + \frac{\sum_{i\in S} p_i \log p_i}{\left(\sum_{i\in S} p_i\right)^2} = \frac{\frac{\sum_{i\in S} p_i\cdot \log p_i}{\sum_{i\in S} p_i} - \log p_j}{\sum_{i\in S} p_i}
        = \frac{\sum_{i\in S} q_i \log p_i - \log p_j}{\sum_{i\in S} p_i},
    \end{align*}
    where $q_i := \frac{p_i}{\sum_{i\in S} p_i}$. 
    Recall that $\tau = \prod_{i\in S} p_i^{q_i}$, i.e., $\log \tau = \sum_{i\in S} q_i \log p_i$. The claim follows.
\end{proof}

Next, we state some simple properties of the weighted geometric mean $\tau$.
\begin{claim}\label{cl:tau-stable}
    The weighted geometric mean $\tau$ over a set $S$ does not change if $\tau$ is added to $S$.
\end{claim}
\begin{proof}
    Recall that $\tau = \prod_{i\in S} p_i^{q_i}$ where $q_i = \nf {p_i}{p}$, where $p = \sum_{i\in S} p_i$. Let $q'_i$ be the weighted geometric mean if we include $\tau$, i.e., 
    \[
        q'_i = \begin{cases}
                    \nf{p_i}{p+\tau} & \text{ if } i\in S \\
                    \nf{\tau}{p+\tau} & \text{ if } p_i = \tau.
                \end{cases}
    \]
    The new weighted geometric mean $\tau'$ is given by
    \begin{align*}
        \tau' 
        &= \left(\prod_{i\in S} p_i^{q'_i}\right) \cdot \tau^{\nf{\tau}{p+\tau}} = \left(\prod_{i\in S} p_i^{\nf{p_i}{p+\tau}}\right) \cdot \tau^{\nf{\tau}{p+\tau}} = \left(\prod_{i\in S} p_i^{\nf{p_i}{p}}\right)^{\nf{p}{p+\tau}} \cdot \tau^{\nf{\tau}{p+\tau}} = \tau^{\nf{p}{p+\tau}} \cdot \tau^{\nf{\tau}{p+\tau}} = \tau. \qedhere
    \end{align*}
\end{proof}

\begin{claim}\label{cl:tau-monotone}
    Consider two sets $S_1 = S \cup \{\pmax\}, S_2 = S \cup \{\pmin\}$ such that for any $i\in S$, we have $\pmin \le p_i \le \pmax$. Let $\tau_1, \tau_2$ be the weighted geometric means of $S_1, S_2$ respectively. Then, $\tau_1 \ge \tau_2$.
\end{claim}
\begin{proof}
    Let $p = \sum_{i\in S} p_i$. We have 
    \begin{align*}
        \frac{\tau_1}{\tau_2} 
        &= \frac{\prod_{i\in S} p_i^{\nf{p_i}{p+\pmax}}}{\prod_{i\in S} p_i^{\nf{p_i}{p+\pmin}}}
        \cdot \frac{\pmax^{\nf{\pmax}{p+\pmax}}}{\pmin^{\nf{\pmin}{p+\pmin}}} \ge \prod_{i\in S} p_i^{\frac{p_i(\pmin-\pmax)}{(p+\pmax)(p+\pmin)}}\cdot \pmax^{\frac{p(\pmax-\pmin)}{(p+\pmax)(p+\pmin)}} = \prod_{i\in S} \left(\frac{\pmax}{p_i}\right)^{\frac{p_i(\pmax-\pmin)}{(p+\pmax)(p+\pmin)}} \ge 1.\qedhere
    \end{align*}
\end{proof}
Now, we are ready to prove \Cref{lem:interval}.
\begin{proof}[Proof of \Cref{lem:interval}]
    Let $\Smax$ denote a subset that maximizes $\ent(\cD_S)$.
    Let $\pmin, \pmax$ respectively denote the minimum and maximum values in $\Smax$. Assume for contradiction that there is some $q \notin \Smax$ such that $\pmax < q < \pmin$. We will show that $\max\left(\ent(\Smax \cup\{q\}\setminus \{\pmax\}), \ent(\Smax \cup\{q\}\setminus \{\pmin\})\right) \ge \ent(\Smax)$. Clearly, this would establish the lemma.

    We consider two continuous processes. The first is to decrease the value $\pmax$. As we do so, it might not remaining the maximum value in set $\Smax$ but we still denote the dynamic value by $\pmax$. Note that the value of the weighted geometric mean of $\Smax$, denoted $\tau$, also changes; the termination condition we define below is with respect to the current value of the geometric mean and not its original value. This process ends when either of the following two events happens:
    \begin{itemize}
        \item either, $\pmax$ decreases to the current value of the weighted geometric mean $\tau$.
        \item or, $\pmax$ decreases all the way to $q$ while always staying above the (dynamic) weighted geometric mean $\tau$. Then, by \Cref{cl:tau-opt}, the value of $\ent(\Smax)$ is non-decreasing throughout this entire process. Therefore, we have replaced $\pmax$ by $q$ and not decreased $\ent(\Smax)$, which is a contradiction.
    \end{itemize}

    Next, we consider the second process. This is similar to the first process, except that we now increase $\pmin$ instead of decreasing $\pmax$. As earlier, as we increase the value of $\pmin$, it might not remaining the minimum value in set $\Smax$ but we still denote the dynamic value by $\pmin$. Also, as in the previous process, the value of the weighted geometric mean $\tau$ changes during this process, and the termination condition we define below is with respect to the current value of the geometric mean and not its original value. This process ends when either of the following two events happens:
    \begin{itemize}
        \item either, $\pmin$ increases to the current value of the weighted geometric mean $\tau$.
        \item or, $\pmin$ increases all the way to $q$ while always staying below the (dynamic) weighted geometric mean $\tau$. Then, by \Cref{cl:tau-opt}, the value of $\ent(\Smax)$ is non-decreasing throughout this entire process. Therefore, we have replaced $\pmin$ by $q$ and not decreased $\ent(\Smax)$, which is a contradiction.  
    \end{itemize}
    We are left with the case that in both processes, the first event happened. Let $\tau_1, \tau_2$ be the final values of $\tau$ in these two processes respectively. For the first event to happen in the first process, we must have $q < \tau_1$; similarly, for the first event to happen in the second process, we must have $q > \tau_2$. Thus 
    \begin{equation}\label{eq:tau}
        \tau_2 < q < \tau_1.
    \end{equation}
    By $\Cref{cl:tau-stable}$, the value of $\tau_1$ is the weighted geometric mean of the set $\Smax\setminus \{\pmax\}$; similarly, the value of $\tau_2$ is the weighted geometric mean of the set $\Smax\setminus \{\pmin\}$. Therefore, by \Cref{cl:tau-monotone}, $\tau_1 < \tau_2$. This contradicts \eqref{eq:tau}.
\end{proof}

\section{Lower Bound on \lru}
\label{sec:lru_lb}

\begin{theorem}
	There are distributions $\cD$ such that the cost of LRU is $\Omega(\tailp) \cdot \hsub$.
\end{theorem}

\begin{proof}
	Let $\cD$ be a distribution that has one page with probability $\eps = 4 / (k \log k)$, and $k$ pages with probability $(1-\eps)/k$. By a Markov bound, the $\eps$ probability page appears by time $2 / \eps = k \log k / 2$ with probability at least $1/2$. Furthermore, by lower bounds on the coupon collector time (see, e.g.,  \cite[pages 60-63]{Motwani95}) the last page of the top $k$ appears before time $9/10 \cdot k \log k$ with probability less than $1/4$. Hence, with probability at least $1/4$, the algorithm pays $(1-\eps)/k = \Omega(\tailp \cdot \hsub)$.
\end{proof}

This does not rule out that LRU is $\hsub$-competitive, since this bound is witnessed by a distribution for which $\tailp \approx \opt$, but it means a different proof technique than ours is required.

{\small
\bibliography{refs}
\bibliographystyle{plainurl}
}

\end{document}